\documentclass[letterpaper, 10 pt, conference]{ieeeconf}  

\IEEEoverridecommandlockouts                              
\usepackage{amsmath}
\usepackage{amssymb}
\usepackage{amsfonts}

\usepackage{amsthm}
\usepackage{comment}
\usepackage[color=green!40]{todonotes}

\theoremstyle{plain}
\newtheorem{assumption}{Assumption}

\newtheorem{lemma}{Lemma}

\theoremstyle{plain}
\newtheorem{theorem}{Theorem}
\theoremstyle{definition}
\newtheorem{definition}{Definition}
\theoremstyle{definition}

\theoremstyle{definition}
\newtheorem{problem}{Problem}

\title{\LARGE \bf
Social Resource Allocation in a Mobility System with Connected and Automated Vehicles: A Mechanism Design Problem
}

\author{Ioannis Vasileios Chremos, \textit{Student Member, IEEE}, and Andreas A. Malikopoulos, \textit{Senior Member, IEEE}%
\thanks{This research was supported in part by ARPAE's NEXTCAR program under the award number DE-AR0000796 and by the Delaware Energy Institute (DEI).}%
\thanks{The authors are with the Department of Mechanical Engineering, University of Delaware, Newark, DE 19716 USA (emails: \tt\small{ichremos@udel.edu}; \tt\small{andreas@udel.edu}.)}%
}

\begin{document}

\maketitle
\thispagestyle{empty}
\pagestyle{empty}

\begin{abstract}

In this paper, we investigate the social resource allocation in an emerging mobility system consisting of connected and automated vehicles (CAVs) using mechanism design. CAVs provide the most intriguing opportunity for enabling travelers to monitor mobility system conditions efficiently and make better decisions. However, this new reality will influence travelers' tendency-of-travel and might give rise to rebound effects, e.g., increased-vehicle-miles traveled. To tackle this phenomenon, we propose a mechanism design formulation that provides an efficient social resource allocation of travel time for all travelers. Our focus is on the socio-technical aspect of the problem, i.e., by designing appropriate socio-economic incentives, we seek to prevent potential rebound effects. In particular, we propose an economically inspired mechanism to influence the impact of the travelers' decision-making on the well-being of an emerging mobility system.

\end{abstract}

\section{Introduction}


Nowadays, it is nearly impossible to commute in a major urban area without the frustration of a traffic jam or congestion. Congestion leads to more accidents and altercations, and, most importantly, congestion is one of the key contributors that damages the environment (e.g., air pollution caused by the huge numbers of idling engines). It is highly expected that emerging mobility systems, e.g., connected and automated vehicles (CAVs), will be able to eliminate congestion and increase mobility efficiency in terms of energy and travel time \cite{Malikopoulos2017}. However, urban social life has been greatly associated with the technological impact of the car, which compels us to reassess the relationship between automobility and social life \cite{sheller,bissell}. Thus, it is vital to study the impact of CAVs in a socio-technical context focusing on the social dynamics.
%
%
The most novel and defining of all the formidable characteristics of the emerging mobility system is its socio-economic complexity. Future mobility systems will enable human-vehicle interaction and allow enhanced and universal accessibility. Evident from similar technological revolutions (e.g., the impact of elevators on building design and social class hierarchies \cite{bernard}), human social perspective and view can have a tremendous effect on how technological innovations are utilized and implemented. Similarly, CAVs are expected to become a socially disruptive innovation with vast technological, commercial, and regulatory implications. For example, in the form of rebound effects, the benefits of convenience and safety could potentially lead people to travel more frequently using their car, and thus, increase the traffic volume in the transportation network. Even though CAVs are not commercially available yet, the motivation behind our research is the following: ``systems with intertwined social and technological dimensions are not guaranteed to exhibit an optimal performance."

In previous work, we modeled the human social interaction with CAVs as a social dilemma in a game-theoretic approach \cite{chremos2020}. We investigated the \emph{social-mobility dilemma}, i.e., the binary decision-making of travelers between commuting with a CAV or using public transportation. In this paper, using mechanism design, we model the routing of travelers in a transportation network with CAVs as a social resource allocation problem.



Mechanism design theory has emerged to mathematically model, analyze, and solve informationally decentralized problems involving systems of multiple rational and intelligent agents \cite{myerson}. Mechanism design is concerned with methodologies that implement system-wide optimal solutions to a myriad of problems - problems in which the strategically interacting agents can hide their true preferences for better individual benefits, thus hurting the overall efficiency of the system. It has been widely used in areas like communication networks \cite{kakhbod}, power markets \cite{silva}, and social networks \cite{dave2020}. A mechanism may be defined as a mathematical structure that models institutions through which economic activity is guided and coordinated \cite{hurwicz}. We are using the notion of a mechanism in this sense, though the economic activity we aim to control is the allocation of travel time among travelers in a mobility system. Our proposed mechanism presupposes a central authority (e.g., central computer) that gathers all routing requests and travel time demands from CAVs around a city with a road infrastructure that supports connected and automated traffic. In this context, we build and design appropriate protocols and interfaces (e.g., tolls, subsidies) for a central traffic management computer, which will guarantee the realization of the desired outcome, i.e., maximizing social welfare and eliminating congestion.


The authors in \cite{hurwicz} formulated a resource allocation problem within the framework of mechanism design. This work led to a spark of research as mechanism design has been used extensively in communication networks in the form of decentralized resource allocation problems \cite{kakhbod,johari}, and also in transportation \cite{vasirani,wenqian}. In such problems, the main methodology is to apply the Vickrey-Clarke-Groves (VCG) mechanisms, which are direct mechanisms that achieve a socially-optimal solution as a dominant strategy. Because the VCG mechanisms have certain limitations (e.g., they are not budget balanced), there have been attempts to use different approaches to solve the mechanism design problem. For example, by adopting the Nash equilibrium (NE) as the solution concept of the mechanism, a surrogate optimization method can be used where the network manager asks the agents to report a bundle of messages that approximate their private information \cite{johari,tansu}.

In this paper, we focus on the social perspective of the emerging mobility systems with CAVs. It is widely accepted that CAVs will revolutionize urban mobility and the way people commute. An example would be for CAVs to make empty trips, i.e., no travelers, to avoid parking, and thus add extra congestion in the network \cite{de_almeida}. In addition, CAVs could potentially affect drivers' behavior and have an impact on traffic performance in general \cite{aria}. The question of the actual impact of CAVs on travel, energy, and carbon demand has attracted considerable attention \cite{wadud}. Depending on different environmental indicators, the authors in \cite{vivanco} provided a practical microeconomic environmental rebound effect model. So far, there has been research on the effects of a considerate penetration of shared CAVs in a major metropolitan area \cite{martinez}. However, most studies on CAVs have focused on how to coordinate CAVs in different traffic scenarios \cite{jackeline,zhao}. 


In this paper, we investigate the travel time provisioning in transportation networks with CAVs and strategic travelers. The main contribution of this paper is the development of an informationally decentralized travel time social allocation mechanism with strategic travelers possessing the following properties: (a) existence of at least one Nash equilibrium (NE), (b) budget balanced at equilibrium, (c) individually rational, (d) strongly implementable at NE, and (e) feasible at or off of equilibrium. Another contribution of the paper is that the design of our mechanism's tolls for the travelers' utilization of the network's resources is intuitive enough to provide a good understanding of the practical implementation of the mechanism.


The remainder of the paper is organized as follows. In Section \ref{section_formulation}, we present the mathematical formulation of our proposed mechanism. In Section \ref{section_specification_md}, we provide its formal specification, and then, in Section \ref{section_properties}, we formally show that our proposed mechanism has properties (a) - (e). Finally, in Section \ref{section_conclusion}, we draw some concluding remarks and discuss potential avenues for future research.

\section{Mathematical Formulation} \label{section_formulation}

We consider a transportation network represented by a graph $\mathcal{G} = (\mathcal{V}, \mathcal{E})$, where $\mathcal{V} = \{1, \dots, V\}$ corresponds to the index set of vertices and $\mathcal{E} = \{1, \dots, E\}$ the index set of directed edges. Each edge $e \in \mathcal{E}$ has a fixed capacity, i.e., $c_e \in \mathbb{R}_{> 0}$, e.g., a high capacity $c_e$ corresponds to a highway while a low capacity corresponds to an urban road. There are $n \in \mathbb{N}_{\geq 2}$ travelers represented by the set $\mathcal{I} = \{1, 2, \dots, n\}$. Each traveler $i$ is associated with an origin-destination pair $(o_i, d_i) \in \mathcal{V} \times \mathcal{V}$. The utilization of the roads in $\mathcal{G}$ is done by the use of CAVs, where each CAV corresponds to one traveler. We consider 100\% penetration rate of CAVs.



\begin{definition}\label{defn_travel_time}
    A traveler $i\in\mathcal{I}$ seeks to commute from $o_i$ to $d_i$ via a given and fixed route $p_i(o_i, d_i)$ at preferred travel time, denoted by $\theta_i \in \Theta_i = [0, + \infty)$. In game theoretic terms, $\theta_i$ is the type of traveler $i$. We denote the type profile of all travelers by $\theta = (\theta_1, \theta_2, \dots, \theta_n)$.
\end{definition}

In addition, each edge $e \in \mathcal{E}$ in the network is characterized by $\underline{\theta} ^ e$ which represents the minimum possible travel time that any traveler can experience if edge $e \in \mathcal{E}$ is an empty (uncongested) road. This allows us to take into account rural or urban roads of different traffic capacities in the transportation network $\mathcal{G}$.




Next, each traveler $i \in \mathcal{I}$ has a cost function $v_i$ which expresses the ``commute-satisfaction" that traveler $i$ experiences from commuting in $(o_i, d_i)$ with travel time $\theta_i$. We expect $v_i$ and $\theta_i$ to be traveler $i$'s private information (i.e., unknown to the network manager).

\begin{assumption}\label{assumption_satisfaction_function}
    Assume that $v_i : \mathbb{R}_{\geq 0} \to \mathbb{R}$ is continuously differentiable, strictly concave, and strictly decreasing in $\theta_i$ with $v_i(0) = 0$.
\end{assumption}



Next, we denote by $t_i$ the monetary payment made by traveler $i$ to the network manager. We have $t_i \in \mathbb{R}$, i.e., a positive $t_i$ means that traveler $i$ pays a toll and a negative $t_i$ means that $i$ receives a monetary subsidy. Thus, in our mechanism, traveler $i$'s total utility is given by
\begin{equation}\label{eqn_utility}
    u_i(\theta_i, t_i) = v_i(\theta_i) - t_i.
\end{equation}

We consider that all travelers are rational and intelligent decision-makers in the system. Each traveler $i \in \mathcal{I}$ has two objectives: (i) to reach their destination, and (ii) to maximize their own utility. A social consequence of the travelers' behavior is that there is an individual disregard of the overall good of the system and it is natural to expect that at least one edge $e \in \mathcal{E}$ will exceed its maximum capacity. If the network manager does not intervene, then congestion is to be expected. So, using appropriate monetary payments, the network manager can incentivize travelers to report truthfully their type $\theta_i$ and allocate travel time on each edge $e \in \mathcal{E}$ in such a way that all travelers are satisfied and congestion is prevented. To achieve this, the network manager's objective is to maximize the overall ``social welfare" of the network and ensure that the network remains congestion-free. The social welfare function is defined as the $\sum_{i \in \mathcal{I}} v_i(\theta_i)$ and denoted by $\mathcal{W}$. We choose to define the social welfare as the sum of the utilities of all travelers because we follow the utilitarian principles, i.e., we measure the collective benefits gained by the travelers in the transportation network.

Next, note that the travelers' strategic behavior indicates a natural competition over the utilization of the edges.

\begin{definition}
    Given $e \in \mathcal{E}$, we define the following sets: (i) the set $\mathcal{S}_e$ of all travelers that edge $e$ is part of their route that connects $o_i$ and $d_i$, and (ii) the set $\mathcal{R}_i$ of traveler $i$'s edges that consist of their route $p_i(o_i, d_i)$.
\end{definition}

Before we continue, we introduce the notion of \emph{reverse value of time}, say parameter $\alpha_i \in \mathbb{R}_{\geq 1}$, that can vary among each traveler $i \in \mathcal{I}$. The social parameter $\alpha_i \in (\underline{\alpha}, \overline{\alpha})$, where $\underline{\alpha} \geq 1$, can be interpreted as follows. If $\alpha_i \to \overline{\alpha},$ traveler $i$ is willing to tolerate a slightly higher travel time, while if $\alpha_i \to \underline{\alpha}$, traveler $i$ is not willing to tolerate a higher travel time. We assume that each traveler $i \in \mathcal{I}$ can be classified based on socio-economic demographic data (e.g., mobility choices and travel tendencies, civil status and income) \cite{brownstone2003}.


\begin{problem}\label{problem_centralized}
    The centralized social-welfare maximization problem is presented below:
        \begin{gather}
            \max_{\theta_i ^ e} \sum_{i \in \mathcal{I}} \sum_{e \in \mathcal{R}_i} v_i(\theta_i ^ e), \notag \\
            \text{subject to: }
            \theta_i ^ e \geq \underline{\theta} ^ e, \quad \forall e \in \mathcal{E}, \quad \forall i \in \mathcal{I}, \label{constraint_nonnegativity} \\
            \sum_{i \in \mathcal{S}_e} \alpha_i \cdot \theta_i ^ e \leq c_e, \quad \forall e \in \mathcal{E} \label{constraint_capacity},
        \end{gather}
    where $\theta_i ^ e$ is the travel time of traveler $i$ on edge $e$ with $\theta_i = \sum_{e \in \mathcal{R}_i} \theta_i ^ e$, and $v_i(\theta_i) = \sum_{e \in \mathcal{R}_i} v_i(\theta_i ^ e)$; inequalities \eqref{constraint_nonnegativity} ensure that each traveler $i$'s travel time $\theta_i ^ e$ on all edges $e \in \mathcal{E}$ is non-negative but not zero at any case; and inequality \eqref{constraint_capacity} expresses the network's capacity on each edge $e \in \mathcal{E}$.
\end{problem}

By Assumption \ref{assumption_satisfaction_function}, it is imperative to impose a network threshold on the feasible values of each traveler $i$'s travel time. We can achieve this in \eqref{constraint_nonnegativity} by only accepting travelers' travel times that are above $\underline{\theta} ^ e$. Also, we interpret $\theta_i = 0$ to be the case of traveler $i$ not seeking to commute instead of wishing to commute in zero time.

Problem \ref{problem_centralized} would be a standard convex optimization problem if the strategic travelers were expected to report their private information truthfully. As this is unreasonable to expect from strategic decision-makers, the network manager in order to solve Problem \ref{problem_centralized} is tasked to elicit the necessary information using monetary incentives.

\subsection{The Mechanism Design Problem}

In our formulation, we use the NE as our solution concept. However, a NE requires complete information. But, we can interpret a NE as the fixed point of an iterative process in an incomplete information setting \cite{reichelstein,ledyard}. This is in accordance with J. Nash's interpretation of a NE, i.e., the complete information NE can be a possible equilibrium of an iterative learning process.

In this section, we present the fundamentals of an indirect and decentralized resource allocation mechanism following the framework presented in \cite{hurwicz}. First, we need to specify a set of messages that all travelers have access and are able to use in order to communicate information. Based on this information, travelers make decisions which affect the reaction of the network manager. Once the communication between the network manager and the travelers is complete, we say that the mechanism induces a game; strategic travelers then compete for the network's resources. In this line of reasoning, we define formally below what we mean by indirect mechanism and induced game.

An indirect mechanism can be described as a tuple of two components, namely $\langle M, g \rangle$. We write $M = (M_1, M_2, \dots, M_n)$, where $M_i$ defines the set of possible messages of traveler $i$. Thus, the travelers' complete message space is $\mathcal{M} = M_1 \times \dots \times M_n$. The component $g$ is the outcome function defined by $g : \mathcal{M} \to \mathcal{O}$ which maps each message profile to the output space $\mathcal{O} = \{(\theta_1, \dots, \theta_n), (t_1, \dots, t_n) \; | \; \theta_i \in \mathbb{R}_{\geq 0}, \; t_i \in \mathbb{R}\}$,
%
%
i.e., the set of all possible travel time allocations to the travelers and the monetary payments (e.g., toll, subsidies) made or received by the travelers. The outcome function $g$ determines the outcome, namely $g(\mu)$ for any given message profile $\mu = (m_1, \dots, m_n) \in \mathcal{M}$. The payment function $t_i : \mathcal{M} \to \mathbb{R}$ determines the monetary payment made or received by a traveler $i \in \mathcal{I}$.

\begin{definition}\label{defn_induced}
    A mechanism $\langle M, g \rangle$ together with the utility functions $(u_i)_{i \in \mathcal{I}}$ induce a game $\langle M, g, (u_i)_{i \in \mathcal{I}} \rangle$, where each utility $u_i$ is evaluated at $g(\mu)$ for each traveler $i \in \mathcal{I}$.
\end{definition}


\begin{definition}\label{defn_ne}
    Consider a game $\langle M, g, (u_i)_{i \in \mathcal{I}} \rangle$. The solution concept of NE is a message profile $\mu ^ *$ such that $u_i (g(m_i ^ *, m_{- i} ^ *)) \geq u_i(g(m_i, m_{- i} ^ *))$, for all $m_i \in M_i$ and for each $i \in \mathcal{I}$, where $m_{- i} = (m_1, \dots, m_{i - 1}, m_{i + 1}, \dots, m_n)$.
\end{definition}



\begin{definition}\label{defn_ir}
    Let the \emph{utility of no participation} of a traveler $i \in \mathcal{I}$ to be given by $u_i(0, 0) = v_i(0) = 0$. Then, we say that a mechanism is individually rational if $u_i (g(\mu ^ *)) \geq 0$, for all $i \in \mathcal{I}$, and all NE $\mu ^ * \in \mathcal{M}$.
\end{definition}

\section{Proposed Mechanism}\label{section_specification_md}

In this section, we show how the network manager can design monetary incentives which achieve the desirable goal, i.e.,  align everyone's decisions by incentivizing them to send social-welfare supporting messages. But first, we need to establish the informational structure of our mechanism. The network manager has complete knowledge of the network's topology and resources and travelers know only their own utility which they report privately to the network manager. Before we continue, let us define explicitly a traveler's message. For each $i \in \mathcal{I}$, message $m_i \in M_i$ is given by $m_i = (\tilde{\theta}_i, \tau_i)$,
%
%
where $\tilde{\theta}_i = (\tilde{\theta}_i ^ e : e \in \mathcal{R}_i)$ is the reported preferred travel time of traveler $i$, and $\tau_i = (\tau_i ^ e : e \in \mathcal{R}_i)$ is the price traveler $i$ is willing to pay for $\tilde{\theta}_i$ along their route.

\begin{definition}\label{defn_others_taxes}
    The average price of all travelers that compete to utilize edge $e \in \mathcal{E}$ other than traveler $i$ is given by $\tau_{- i} ^ e = \sum_{j \in \mathcal{S}_e : j \neq i} \frac{\tau_j ^ e}{|\mathcal{S}_e| - 1}$.
%
%
\end{definition}

Next, for each traveler $i$ and for each edge $e \in \mathcal{E}$ of their route, we endow a fair share for each edge $e \in \mathcal{E}$, i.e., $c_e / |\mathcal{S}_e|$. This can help us design the monetary payments that each traveler is asked to pay. Using Definition \ref{defn_others_taxes}, we propose the following payments, for a particular edge $e \in \mathcal{E}$,
\begin{multline}\label{eqn_tax}
    t_i ^ e(\mu) = \tau_{- i} ^ e \cdot \left(\alpha_i \cdot \tilde{\theta}_i ^ e - \frac{c_e}{|\mathcal{S}_e|}\right) + (\tau_i ^ e - \nu_e) ^ 2 \\
    + \tau_{- i} ^ e \cdot (\tau_i ^ e - \tau_{- i} ^ e) \cdot \left(c_e - \sum_{i \in \mathcal{S}_e} \alpha_i \cdot \tilde{\theta}_i ^ e\right) ^ 2.
\end{multline}
The first term in \eqref{eqn_tax} is the monetary payments (e.g., toll, subsidies) made or received by traveler $i$ corresponding to their travel time allocation $\tilde{\theta}_i ^ e$ on edge $e \in \mathcal{E}$. Intuitively, this means that traveler $i$ will pay a toll that is determined by the other travelers' recommendations and only for the excess of the fair share of travelers over a particular edge. Using this formulation, there is no incentive for traveler $i$ to lie in an attempt to reduce their payment to the network. The second term in \eqref{eqn_tax} corresponds to a penalty that traveler $i$ will pay if she reports a different price $\tau_i ^ e$ from $\nu_e$, where $\nu_e$ represents the Lagrange multiplier corresponding to the capacity constraint defined formally next. The third term in \eqref{eqn_tax}, collectively incentivizes all travelers to bid the same price per unit of travel time and to utilize the full capacity of each edge $e \in \mathcal{E}$.


Thus, given any message profile $\mu$, the total monetary payment $t_i(\mu)$ for traveler $i$ is
\begin{equation}
    t_i(\mu) = \sum_{e \in \mathcal{R}_i} t_i ^ e(\mu) + \phi_i(\tilde{\theta}_i),
\end{equation}
where $\phi_i$ is a monetary incentive that encourages traveler $i$ to report a reasonable travel time demand respecting road rules and the network's efficiency goals. In detail, we have
\begin{equation}\label{eqn_penalty}
    \phi_i(\tilde{\theta}_i) = 
        \begin{cases}
            \gamma, & \; \exists e \in \mathcal{R}_i, \text{ s.t. } \tilde{\theta}_i ^ e > \underline{\theta} ^ e \text{ and } |\mathcal{S}_e| = 1, \\
            0, & \; \exists e \in \mathcal{R}_i, \text{ s.t. } \tilde{\theta}_i ^ e > \underline{\theta} ^ e \text{ and } |\mathcal{S}_e| \geq 2, \\
            \delta, & \; \exists e \in \mathcal{R}_i, \text{ s.t. } \tilde{\theta}_i ^ e = \underline{\theta} ^ e \text{ and } |\mathcal{S}_e| \geq 2, \\
            0, & \; \exists e \in \mathcal{R}_i, \text{ s.t. } \tilde{\theta}_i ^ e = \underline{\theta} ^ e \text{ and } |\mathcal{S}_e| = 1,
        \end{cases}
\end{equation}
where $\gamma, \delta \in \mathbb{R}_{> 0}$ represent the imposition of very high penalties. It is necessary to impose such penalties since for the first case in \eqref{eqn_penalty}, traveler $i$ violates the goal of efficiency in the network and for the third case in \eqref{eqn_penalty}, traveler $i$ violates the goal of road safety. In the severe case of $\tilde{\theta}_i ^ e < \underline{\theta} ^ e$, we have $\phi_i(\tilde{\theta}_i) = + \infty$.

\section{Properties of the Mechanism}\label{section_properties}

In this section, we present the properties of our proposed mechanism.

\begin{lemma}
    Problem \ref{problem_centralized} has a unique optimal solution.
\end{lemma}

\begin{proof}
    The objective function of Problem \ref{problem_centralized} is a sum of several strictly concave functions. Hence, it is strictly concave. Thus, the necessary KKT conditions are also sufficient for optimality. Since the feasible region is non-empty, convex, and compact, we conclude that Problem \ref{problem_centralized} has always a unique optimal solution.
\end{proof}

\begin{lemma}\label{lemma_kkt}
    A solution to Problem \ref{problem_centralized} is unique and optimal if, and only if, it satisfies the feasibility conditions of Problem \ref{problem_centralized} and there exist Lagrange multipliers $\lambda = (\lambda_i ^ e : e \in \mathcal{E})_{i \in \mathcal{I}}$ and $\nu = (\nu_e)_{e \in \mathcal{E}}$ that satisfy the following conditions:
        \begin{align}\label{eqn_kkt_1}
            \frac{\partial v_i({\theta_i ^ e}{} ^ *)}{\partial \theta_i ^ e} + {\lambda_i ^ e}{} ^ * - \sum_{e \in \mathcal{R}_i} \alpha_i \cdot \nu_e ^ * = 0, \\
            {\lambda_i ^ e}{} ^ * \cdot ({\theta_i ^ e}{} ^ * - \underline{\theta} ^ e) = 0, \quad \forall e \in \mathcal{E}, \quad \forall i \in \mathcal{I}, \label{eqn_kkt_2} \\
            \nu_e ^ * \cdot \left(\sum_{i \in \mathcal{S}_e} \alpha_i \cdot {\theta_i ^ e}{} ^ * - c_e\right) = 0, \quad \forall e \in \mathcal{E}, \label{eqn_kkt_3} \\
            {\lambda_i ^ e}{} ^ *, \nu_e ^ * \geq 0, \quad \forall e \in \mathcal{E}, \quad \forall i \in \mathcal{I}. \label{eqn_kkt_4}
        \end{align}
\end{lemma}

\begin{proof}
    First, let us derive the Lagrangian of Problem \ref{problem_centralized}:
        \begin{multline}\label{lagrangian}
            \mathcal{L}(\theta, \lambda, \nu) = \sum_{i \in \mathcal{I}} \sum_{e \in \mathcal{R}_i} v_i(\theta_i ^ e) + \sum_{i \in \mathcal{I}} \sum_{e \in \mathcal{E}} \lambda_i ^ e \cdot (\theta_i ^ e - \underline{\theta} ^ e) \\
            - \sum_{e \in \mathcal{E}} \nu_e \cdot \left(\sum_{i \in \mathcal{S}_e} \alpha_i \cdot \theta_i ^ e - c_e\right).
        \end{multline}
    From \eqref{lagrangian}, it is easy to derive the KKT conditions, i.e.,
        \begin{align}
            \frac{\partial v_i({\theta_i ^ e}{} ^ *)}{\partial \theta_i ^ e} + {\lambda_i ^ e}{} ^ * - \sum_{e \in \mathcal{R}_i} \alpha_i \cdot \nu_e ^ * = 0, \\
            {\lambda_i ^ e}{} ^ * \cdot ({\theta_i ^ e}{} ^ * - \underline{\theta} ^ e) = 0, \quad \forall e \in \mathcal{E}, \quad \forall i \in \mathcal{I}, \\
            \nu_e ^ * \cdot \left(\sum_{i \in \mathcal{S}_e} \alpha_i \cdot {\theta_i ^ e} ^ * - c_e\right) = 0, \quad \forall e \in \mathcal{E}, \\
            {\lambda_i ^ e}{} ^ *, \nu_e ^ * \geq 0, \quad \forall e \in \mathcal{E}, \quad \forall i \in \mathcal{I}.
        \end{align}
    Since the KKT conditions are necessary and sufficient to guarantee the optimality of any allocation of travel time that satisfies them, it is enough to find ${\lambda_i ^ e}{} ^ *$ and $\nu_e ^ *$ such that the above conditions are satisfied.
\end{proof}

\begin{theorem}[Feasibility]\label{thm_feasibility}
     For any message profile $\mu$, the corresponding travel time allocation $\theta$ is a feasible point of Problem \ref{problem_centralized}.
\end{theorem}

\begin{proof}
    Consider any traveler $i$ and denote by $\mathcal{C}$ the constraint set of Problem \ref{problem_centralized}. Then, for a reported preferred travel time $\tilde{\theta}_i$, the travel time $\theta_i$ of Problem \ref{problem_centralized} generated by the outcome function is equal to (i) $\tilde{\theta}_i$ if $\tilde{\theta}_i \in \mathcal{C}$; or (ii) $\theta_i ^ 0$ if $\tilde{\theta}_i \notin \mathcal{C}$,
    where $\tilde{\theta}_i = (\tilde{\theta}_i ^ e : e \in \mathcal{R}_i)$, and $\theta_i ^ 0$ is the point on the boundary of $\mathcal{C}$ (i.e., we ignore the ``unreasonable" demand of traveler $i$ and allocate only the portion of the resource that is available). By construction, it follows immediately that if $\tilde{\theta}_i \in \mathcal{C}$, then the allocation $\theta_i$ is feasible for any traveler $i \in \mathcal{I}$. In the case of $\tilde{\theta}_i \notin \mathcal{C}$, the allocation is on the boundary of $\mathcal{C}$, hence it is still feasible as the constraint set of Problem \ref{problem_centralized} is closed. Thus, the result follows.
\end{proof}

\begin{lemma}\label{lemma_lagrange_prices}
    Let $\mu ^ *$ be a NE of the induced game. Then, we have ${\tau_i ^ e} ^ * = \nu_e ^ *$, for all $i \in \mathcal{I}$ and each $e \in \mathcal{R}_i$. In addition, it follows that $\tau_{- i} ^ e = \sum_{j \in \mathcal{S}_e : j \neq i} \frac{\tau_j ^ e}{|\mathcal{S}_e| - 1} = {\tau_i ^ e} ^ *$.
\end{lemma}

\begin{proof}
    Suppose there is one traveler, say $i$, that deviates from the NE message profile $\mu ^ *$ and instead reports the message $m_i = (\tilde{\theta}_i ^ *, \tau_i)$. This deviation to be justifiable has to provide a higher utility to traveler $i \in \mathcal{I}$. But, we have
        \begin{multline}\label{eqn_lagrange_1}
            v_i(\tilde{\theta}_i ^ *) - t_i(m_i ^ *, m_{- i} ^ *) \geq v_i(\tilde{\theta}_i ^ *) - t_i(m_i, m_{- i} ^ *).
        \end{multline}
    Next, we substitute \eqref{eqn_tax} into \eqref{eqn_lagrange_1}. For ease of notational exposition, let $\xi = \left(c_e - \sum_{i \in \mathcal{S}_e} \alpha_i \cdot {\tilde{\theta}_i ^ e}{} ^ *\right) ^ 2$. Thus,
        \begin{multline}\label{equal_prices}
            \sum_{e \in \mathcal{R}_i} ({\tau_i ^ e} ^ * - \nu_e ^ *) ^ 2 + {\tau_{- i} ^ e} ^ * \cdot ({\tau_i ^ e} ^ * - {\tau_{- i} ^ e} ^ *) \cdot \xi \\
            \leq \sum_{e \in \mathcal{R}_i} (\tau_i ^ e - \nu_e ^ *) ^ 2 + {\tau_{- i} ^ e} ^ * \cdot (\tau_i ^ e - {\tau_{- i} ^ e} ^ *) \cdot \xi.
        \end{multline}
    Since traveler $i$ behaves as a utility-maximizer, we need to minimize the right hand side of \eqref{equal_prices}. Thus, the best price is $\tau_i ^ e = {\tau_{- i} ^ e} ^ *$, and also the solution of the minimization problem $\min_{(\tau_i ^ e)} \sum_{e \in \mathcal{R}_i} (\tau_i ^ e - \nu_e ^ *) ^ 2$. Therefore, at $\mu ^ *$, we have ${\tau_i ^ e} ^ * = \nu_e ^ *$, for all $e \in \mathcal{E}$ and for all $i \in \mathcal{I}$ and $\tau_{- i} ^ e = \sum_{j \in \mathcal{S}_e : j \neq i} \frac{\tau_j ^ e}{|\mathcal{S}_e| - 1} = {\tau_i ^ e} ^ *$ follows immediately.
\end{proof}

\begin{lemma}
    Let $\mu ^ *$ be a NE of the induced game. Then, for every traveler $i \in \mathcal{I}$, we have $\phi_i(\tilde{\theta}_i ^ *) = 0$.
\end{lemma}

\begin{proof}
    We prove this by contradiction. Suppose there exists a NE message $\mu ^ * = (m_i ^ * = (\tilde{\theta}_i ^ *, \tau_i ^ *))_{i \in \mathcal{I}}$ such that $\phi_i(\tilde{\theta}_i ^ *) \neq 0$ for traveler $i \in \mathcal{I}$. By \eqref{eqn_penalty}, we only have two cases to consider: let ${\tilde{\theta}_i ^ e}{} ^ * > \underline{\theta} ^ e$ with $|\mathcal{S}_e| = 1$ (the proof for the other case is similar). Suppose traveler $i$ deviates from the NE with message $m_i = ((\tilde{\theta}_i ^ e = \underline{\theta} ^ e : e \in \mathcal{R}_i), \tau_i ^ *)$. By Definition \ref{defn_ne}, we have
        \begin{equation}\label{eqn_penalty_1}
            u_i(g(m_i, m_{- i} ^ *)) \leq u_i(g(m ^ *)).
        \end{equation}
    Substitute \eqref{eqn_utility}, \eqref{eqn_tax}, and \eqref{eqn_penalty} into \eqref{eqn_penalty_1} and then Lemma \ref{lemma_lagrange_prices} gives
        \begin{equation}\label{eqn_penalty_2}
            [v_i(\tilde{\theta}_i) - v_i({\tilde{\theta}_i}{} ^ *)] - \sum_{e \in \mathcal{S}_e} \alpha_i \cdot \nu_e ^ * \cdot (\tilde{\theta}_i ^ e - {\tilde{\theta}_i ^ e}{} ^ *) + \phi_i(\tilde{\theta}_i ^ *) \leq 0,
        \end{equation}
    where by Assumption \ref{assumption_satisfaction_function}, the first difference term of \eqref{eqn_penalty_2} is negative; likewise the difference of $(\tilde{\theta}_i ^ e - {\tilde{\theta}_i ^ e}{} ^ *)$ is positive. Thus, it follows that, since $\phi_i(\tilde{\theta}_i ^ *) \gg 0$, traveler $i$ rightfully deviates from the NE $\mu ^ * = (m_i ^ * = (\tilde{\theta}_i ^ *, \tau_i ^ *))_{i \in \mathcal{I}}$ such that $\phi_i(\tilde{\theta}_i ^ *) \neq 0$ as \eqref{eqn_penalty_2} cannot be true (by construction of \eqref{eqn_penalty}). Since the case of $\tilde{\theta}_i ^ e < \underline{\theta} ^ e$, where $\phi_i(\tilde{\theta}_i) = + \infty$ is straightforward to show, and the proof is complete.
\end{proof}

\begin{theorem}[Budget Balance]
    Let the message profile $\mu ^ *$ be a NE of the induced game. The proposed mechanism at $\mu ^ *$ does not require any external or internal monetary payments, i.e., $\sum_{i \in \mathcal{I}} t_i(\mu ^ *) = 0$ for all $\mu ^ *$.
\end{theorem}

\begin{proof}
    Summing \eqref{eqn_tax} over all travelers yields $\sum_{i \in \mathcal{I}} t_i(\mu ^ *) = \sum_{i \in \mathcal{I}} \left[\sum_{e \in \mathcal{R}_i} t_i ^ e (\mu ^ *)\right] = \sum_{e \in \mathcal{H}} \sum_{i \in \mathcal{S}_e} t_i ^ e (\mu ^ *)$,
    where $\mathcal{H}$ is the set of competitive edges in the network (i.e., any edge utilized by more than two travelers). Hence,
        \begin{multline}
            \sum_{e \in \mathcal{H}} \sum_{i \in \mathcal{S}_e} {\tau_{- i} ^ e} ^ * \cdot \left(\alpha_i \cdot {\tilde{\theta}_i ^ e}{} ^ * - \frac{c_e}{|\mathcal{S}_e|}\right) + ({\tau_i ^ e} ^ * - \nu_e ^ *) ^ 2 \\
            + {\tau_{- i} ^ e} ^ * \cdot ({\tau_i ^ e} ^ * - {\tau_{- i} ^ e} ^ *) \cdot \left(c_e - \sum_{i \in \mathcal{S}_e} \alpha_i \cdot {\tilde{\theta}_i ^ e}{} ^ *\right).
        \end{multline}
    By Lemma \ref{lemma_lagrange_prices}, we have for all $e \in \mathcal{H}$, $\sum_{i \in \mathcal{I}} t_i(\mu ^ *) = \sum_{e \in \mathcal{H}} \nu_e ^ * \cdot \left(\sum_{i \in \mathcal{S}_e} \alpha_i \cdot {\tilde{\theta}_i ^ e}{} ^ * - c_e\right)$,
    which is equal to zero by the KKT conditions in Lemma \ref{lemma_kkt}.
\end{proof}

\begin{theorem}[Individually Rational]
    The proposed mechanism is individually rational. In particular, each traveler prefers the outcome of any NE of the induced game to the outcome of no participation.
\end{theorem}

\begin{proof}
    Let the message profile $\mu ^ *$ be an arbitrary NE of the induced game. We need to show that $u_i(\mu ^ *) \geq u_i(0) = 0$ for each traveler $i$ (see Definition \ref{defn_ir}). Consider the message $m_i = (\tilde{\theta}_i, \tau_i)$ with $\tilde{\theta}_i = 0$ and $\tau_i = (\tau_i ^ e = \nu_e : e \in \mathcal{R}_i)$. That is, traveler $i$ deviates with $m_i$ while the other travelers adhere to the NE $\mu ^ *$. By Definition \ref{defn_ne}, we have the following:
        \begin{align}
            u_i(g(\mu ^ *)) & \geq u_i(g(m_i, m_{- i} ^ *)) \notag \\
            & = v_i(0) - \sum_{e \in \mathcal{R}_i} {\tau_{- i} ^ e}{} ^ * \cdot \left(0 - \frac{c_e}{|\mathcal{S}_e|}\right) \notag \\
            & = \sum_{e \in \mathcal{R}_i} \nu_e ^ * \cdot \left(\frac{c_e}{|\mathcal{S}_e|}\right) \geq 0. \label{eqn_ir}
        \end{align}
    Thus, from \eqref{eqn_ir}, the result follows.
\end{proof}

In our next result, we show that our mechanism is strongly implementable at NE. Strong implementation ensures that the efficient allocation of travel time to the travelers is implemented by all equilibria of the induced game \cite{mathevet}.

\begin{theorem}[Strong Implementation]\label{thm_strong_implementation}
    At an arbitrary NE $\mu ^ *$ of the induced game, the allocation travel time $(\tilde{\theta}_i ^ *)_{i \in \mathcal{I}}$ is equal to the optimal solution $(\theta_i ^ *)_{i \in \mathcal{I}}$ of Problem \ref{problem_centralized} for each $i \in \mathcal{I}$.
\end{theorem}


\begin{proof}
    Suppose $\mu ^ *$ is a NE of the induced game. Then, by Lemma \ref{lemma_lagrange_prices}, it follows that ${\tau_i ^ e} ^ * = {\tau_{- i} ^ e} ^ * = \nu_e ^ *$ for each $e \in \mathcal{R}_i$. Next, consider some traveler $i$ that participates in the mechanism and has preferred travel time $\theta_i$.  The utility of traveler $i$ for such an allocation is given by
        \begin{equation}\label{eqn_implementation_1}
            u_i(g(m_i, m_{- i} ^ *)) = v_i(\theta_i) - t_i(m_i, m_{- i} ^ *),
        \end{equation}
    where $t_i(m_i, m_{- i} ^ *) = \sum_{e \in \mathcal{R}_i} \nu_e ^ * \left(\alpha_i \cdot \theta_i ^ e - \frac{c_e}{|\mathcal{S}_e|}\right)$. By Definition \ref{defn_ne}, it follows that at NE no traveler should have an incentive to deviate. Hence, the maximization of traveler $i$'s utility \eqref{eqn_implementation_1} must be attained at the NE travel time allocation, i.e., $\theta_i ^ * = \tilde{\theta}_i ^ *$. The Nash-maximization problem is
        \begin{equation}\label{eqn_implementation_2}
            {\tilde{\theta}_i ^ e}{} ^ * = \arg \max_{\theta_i ^ e} \left[\sum_{e \in \mathcal{R}_i} v_i(\theta_i ^ e) - \sum_{e \in \mathcal{R}_i} \nu_e ^ * \left(\alpha_i \cdot \theta_i ^ e - \frac{c_e}{|\mathcal{S}_e|}\right)\right],
        \end{equation}
    subject to the exact same constraints as in Problem \ref{problem_centralized}. Now, it is easy to derive the KKT conditions that will give the optimal ``Nash solution." By Lemma \ref{lemma_kkt}, the KKT conditions are necessary and sufficient to guarantee the optimality of any travel time allocation $(\theta_i)_{i \in \mathcal{I}}$ that satisfies them. Thus, it is sufficient to show that there exist appropriate Lagrange multipliers ${\lambda_i ^ e}{} ^ *$ and $\nu_e ^ *$ such that \eqref{eqn_kkt_1} - \eqref{eqn_kkt_3} are satisfied. By setting $\lambda_i ^ e = 0$ and $\nu_e = \tau_i ^ e$ for all $e \in \mathcal{E}$, by differentiation of \eqref{eqn_tax} with respect to $\theta_i ^ e$ and $\tau_i ^ e$, we get
        \begin{gather}\label{eqn_implementation_3}
            \frac{\partial v_i({\tilde{\theta}_i ^ e}{} ^ *)}{\partial \tilde{\theta}_i ^ e} = \sum_{e \in \mathcal{R}_i} \alpha_i \cdot \nu_e ^ *, \quad \forall i \in \mathcal{I}, \\
            \nu_e ^ * \cdot \left(\sum_{i \in \mathcal{S}_e} \alpha_i \cdot {\tilde{\theta}_i ^ e}{} ^ * - c_e\right) = 0, \quad \forall e \in \mathcal{E}. \label{eqn_implementation_4}
        \end{gather}
    It is straightforward to see that \eqref{eqn_implementation_3} and \eqref{eqn_implementation_4} are identical to \eqref{eqn_kkt_1} and \eqref{eqn_kkt_3}, respectively. Condition \eqref{eqn_kkt_2} in both problems holds trivially. Consequently, the solution $\tilde{\theta} ^ * = (\tilde{\theta}_1 ^ *, \dots, \tilde{\theta}_n ^ *)$ of \eqref{eqn_implementation_3} and \eqref{eqn_implementation_4} along with the specification of the payment functions \eqref{eqn_tax} are equivalent to the optimal unique solution of Problem \ref{problem_centralized}. Thus, at any NE $\mu ^ *$, we get an identical allocation $g(\mu ^ *) = (\tilde{\theta}_1 ^ *, \dots, \tilde{\theta}_n ^ *, t_1 ^ *, \dots, t_n ^ *)$ that is equal to the optimal solution of Problem \ref{problem_centralized}, and the proof is complete. 
\end{proof}

\begin{theorem}[Existence]
    Let $\theta ^ *$ be the optimal solution of Problem \ref{problem_centralized} and $\nu_e ^ *$ be the corresponding Lagrange multipliers of the KKT conditions. If for each $i \in \mathcal{I}$, $m_i ^ * = (\tilde{\theta}_i ^ * = \theta_i ^ *, \tau_i ^ *)$, where $\tau_i ^ * = ({\tau_i ^ e}{} ^ * = \nu_e ^ * : \forall e \in \mathcal{R}_i)$ and $\phi_i(\tilde{\theta}_i ^ *) = 0$ for all $i \in \mathcal{I}$. Then the message $\mu ^ * = (m_i ^ *)_{i \in \mathcal{I}}$ is a NE of the induced game.
\end{theorem}

\begin{proof}

    We show that the message profile $\mu ^ *= (m_i ^ *)_{i \in \mathcal{I}}$ where $m_i ^ * = (\tilde{\theta}_i = \theta_i ^ *, \tau_i ^ *)$ is a NE. By Lemma \ref{lemma_kkt}, it follows that $\theta ^ *$ along with the appropriate Lagrange multipliers satisfies the KKT conditions of Problem \ref{problem_centralized} and is the only feasible allocation. For any traveler $i$, the utility at message $\mu ^ *$ is $u_i(g(\mu ^ *)) = v_i(\tilde{\theta}_i ^ *) - \sum_{e \in \mathcal{R}_i} \nu_e ^ * \cdot \left(\alpha_i \cdot {\tilde{\theta}_i ^ e}{} ^ * - \frac{c_e}{|\mathcal{S}_e|}\right)$.
    Now, suppose traveler $i$ deviates from $\mu ^ *$ by changing their message while all the other travelers adhere to the message $\mu ^ *$ (though we would still have ${\tau_{- i} ^ e}{} ^ * = \nu_e ^ *$). We have
        \begin{gather}
            u_i(g(m_i, m_{- i} ^ *)) \leq v_i(\tilde{\theta}_i ') - \sum_{e \in \mathcal{R}_i} \nu_e ^ * \cdot \left(\alpha_i \cdot \tilde{\theta}_i ^ e{} ' - \frac{c_e}{|\mathcal{S}_e|}\right) \notag \\
            \leq \max_{\tilde{\theta}_i '} \left[v_i(\theta_i ') - \sum_{e \in \mathcal{R}_i} \nu_e ^ * \cdot \left(\alpha_i \cdot \tilde{\theta}_i ^ e{} ' - \frac{c_e}{|\mathcal{S}_e|}\right)\right]. \label{eqn_existence}
        \end{gather}
    The maximization problem \eqref{eqn_existence} is equivalent to \eqref{eqn_implementation_2}. As the message $\mu ^ *$ clearly satisfies the KKT conditions of \eqref{eqn_implementation_2}, we have $\theta_i = \tilde{\theta}_i = \tilde{\theta}_i '$, which in turn implies:
        \begin{equation}\label{eqn_existence_last}
            u_i(g(m_i, m_{- i} ^ *)) \leq v_i(\tilde{\theta}_i ^ *) - \sum_{e \in \mathcal{R}_i} \nu_e ^ * \cdot \left(\alpha_i \cdot {\tilde{\theta}_i ^ e} - \frac{c_e}{|\mathcal{S}_e|}\right),
        \end{equation}
    where the right hand side of \eqref{eqn_existence_last} is equal to $u_i(g(\mu ^ *))$, for all $m_i ^ *$ and all $i \in \mathcal{I}$. Therefore, message $\mu ^ *$ is a NE.
\end{proof}

\section{Concluding Remarks}\label{section_conclusion}

In this paper, we formulated the routing of strategic travelers that use CAVs in a transportation network as a social resource allocation mechanism design problem. Considering a Nash-implementation approach, we showed that our proposed informationally decentralized mechanism efficiently allocates travel time to all travelers that seek to commute in the network. Our mechanism induces a game which at least one equilibrium prevents congestion (a significant rebound effect), while also attaining the properties of individually rationality, budget balanced, strongly implementability. Ongoing work includes conducting a simulation-based analysis under different traffic scenarios to showcase the practical implications of our mechanism. Extending and enhancing the traveler-behavioral model, motivated by a social-mobility survey can be a worthwhile undertaking as a future research direction allowing the study of the relationship of emerging mobility and the intricacies of human decision-making.


\addtolength{\textheight}{-12cm}   


\bibliographystyle{IEEEtran}
\bibliography{references}

\end{document}